\documentclass[letterpaper,11pt,onecolumn,accepted=2022-01-10]{quantumarticle}
\pdfoutput=1

\title{On the Hardness of PAC-learning Stabilizer States with Noise}
\author{Aravind Gollakota}
\email{aravindg@cs.utexas.edu}
\author{Daniel Liang}
\email{dliang@cs.utexas.edu}
\affiliation{Department of Computer Science \\ University of Texas at Austin}
\date{January 31, 2022}

\usepackage{mathtools}
\usepackage{amsfonts}
\usepackage{enumerate}
\usepackage{amsthm}
\usepackage{braket}
\usepackage[utf8]{inputenc}
\usepackage[english]{babel}
\usepackage[T1]{fontenc}
\usepackage{amsmath}

\usepackage[pagebackref,colorlinks,citecolor=blue,bookmarks=true]{hyperref}
\usepackage[capitalize]{cleveref}
\usepackage[numbers,sort&compress]{natbib}

\theoremstyle{plain}
\newtheorem{theorem}{Theorem}[section]
\newtheorem{lemma}[theorem]{Lemma}
\newtheorem{corollary}[theorem]{Corollary}
\newtheorem{proposition}[theorem]{Proposition}

\theoremstyle{definition}
\newtheorem{definition}[theorem]{Definition}

\renewcommand{\th}{^\text{th}}
\newcommand{\R}{\mathbb{R}}

\newcommand{\cC}{\mathcal{C}}
\newcommand{\cP}{\mathcal{P}}
\newcommand{\F}{\mathcal{F}}
\newcommand{\E}{\mathcal{E}}
\newcommand{\Epauli}{\mathcal{E}_\text{Pauli}}

\newcommand{\X}{\mathcal{X}}

\let\phi\varphi
\DeclarePairedDelimiter{\inn}{\langle}{\rangle}

\DeclarePairedDelimiterX{\ketbra}[2]{|}{|}{%
	#1\rangle\langle#2%
}
\DeclareMathOperator*{\ex}{\mathbb{E}}
\DeclareMathOperator*{\pr}{\mathbb{P}}
\DeclareMathOperator{\trc}{Tr}
\DeclareMathOperator{\poly}{poly}

\DeclareMathOperator{\sda}{SDA}
\DeclareMathOperator{\sgn}{sign}

\let\hat\widehat

\begin{document}

\maketitle

\begin{abstract}
	We consider the problem of learning stabilizer states with noise in the Probably Approximately Correct (PAC) framework of \citet{aaronson2007learnability} for learning quantum states. In the noiseless setting, an algorithm for this problem was recently given by \citet{rocchetto2018stabiliser}, but the noisy case was left open. Motivated by approaches to noise tolerance from classical learning theory, we introduce the Statistical Query (SQ) model for PAC-learning quantum states, and prove that algorithms in this model are indeed resilient to common forms of noise, including classification and depolarizing noise. We prove an exponential lower bound on learning stabilizer states in the SQ model. Even outside the SQ model, we prove that learning stabilizer states with noise is in general as hard as Learning Parity with Noise (LPN) using classical examples. Our results position the problem of learning stabilizer states as a natural quantum analogue of the classical problem of learning parities: easy in the noiseless setting, but seemingly intractable even with simple forms of noise.
\end{abstract}

\section{Introduction}
A fundamental task in quantum computing is that of learning a description of an unknown quantum state $\rho$. Traditionally this is formalized as the problem of quantum state tomography, where we are granted the ability to form multiple copies of $\rho$ and take arbitrary measurements, and must learn a state $\sigma$ that is close to $\rho$ in trace distance. In an influential work, \citet{aaronson2007learnability} introduced the ``Probably Approximately Correct'' (PAC) framework from computational learning theory \cite{valiant1984theory} as an alternative perspective on this problem. Here the key innovation is that instead of learning $\rho$ in an absolute metric (such as trace distance), we only wish to learn it with respect to a pre-specified distribution on measurements. This requirement is considerably weaker than that of full tomography.

Concretely, let $\E$ denote the space of two-outcome $n$-qubit measurements $E$, where $E$ (corresponding to the POVM $\{E, I - E\}$) accepts a state $\rho$ with probability $\trc(E \rho)$ and rejects it otherwise. Let $D$ be a distribution on $\E$. We are given the ability to sample and perform random measurements from $D$, i.e.\ to form a training set $\{(E_i, Y_i)\}_{i \in [m]}$ where each $E_i \sim D$ and $Y_i$ is the outcome of measuring $\rho$ using $E_i$. Our goal is to learn a state $\sigma$ such that with high probability, \[ \ex_{E \sim D} \left[ (\trc(E \sigma) - \trc(E \rho))^2 \right] \leq \epsilon. \] Usually we also have some basic prior information about $\rho$, such as knowledge of a class $\F$ such that $\rho \in \F$. In the terminology of learning theory, this would be called ``learning the class $\F$.'' It is important to stress that in this framework, while the data \emph{arises from} a quantum state, it is entirely classical in form and representation.

When $\rho$ is an arbitrary $n$-qubit state, \cite{aaronson2007learnability} showed that $O(n)$ ``training examples'', i.e.\ measurement-outcome pairs $(E, \trc(E \rho))$,\footnote{Or more realistically $(E, Y)$, where $Y$ is the random outcome of measuring $\rho$ using $E$.} suffice to statistically determine a state $\sigma$ that approximates $\rho$ in this sense (suppressing the dependence on $\epsilon$ for simplicity). This is in contrast to full state tomography, which requires an exponential number of measurements or copies of $\rho$ \cite{o2016efficient, Haah_2017}. Note, however, that here $\sigma$ is determined purely in a statistical or information-theoretic sense, and for a long time no efficient algorithms were known for actually computing $\sigma$ in settings of interest.

Recently, \citet{rocchetto2018stabiliser} gave an efficient algorithm for the case where $\rho$ is known to be a stabilizer state, and $D$ is any distribution on Pauli measurements. The stabilizer states are an important class of states in quantum computing due to their efficient classical simulability \cite{gottesman1998heisenberg, aaronson2004improved} and their foundational role in quantum error-correction \cite{gottesman1997stabilizer}. A major question left open by \cite{rocchetto2018stabiliser} is: are stabilizer states also efficiently learnable in noisy settings?

Motivated by this question, we introduce to the quantum setting a well-known tool for noise-resilient classical PAC learning, the \emph{statistical query} (SQ) model, and define the problem of SQ-learning quantum states. In this model, rather than receiving labeled measurement-outcome examples of the form $(E, \trc(E\rho))$, the learner is only allowed to make statistical queries to an oracle, and otherwise its goal remains the same. A statistical query is described by a function $\phi : \E \times \{-1, 1\} \to [-1, 1]$ and a tolerance $\tau > 0$, and the oracle responds to the query with $\ex[\phi(E, Y)] \pm \tau$, where the expectation is taken over the random draw of $E \sim D$ and $Y$, the random outcome of measuring $\rho$ using $E$. One can think of this as modeling an experimental setup that is unable to report individual measurement outcomes, but is nevertheless able to estimate expectation values to any desired accuracy. Importantly, an algorithm that is able to work in this restricted setting automatically gains tolerance to several kinds of noise.

The SQ model was originally introduced by \citet{kearns1998efficient} in the setting of Boolean function classes, and has since grown into a highly influential model (see \cite{feldman2016, reyzin2020statistical} for surveys). The model is known to have the following properties: \begin{itemize}
	\item It is a natural restriction of the PAC model that nevertheless captures most known PAC algorithms for a wide range of common classes \cite{hellerstein2007pac, reyzin2020statistical}.  
	\item SQ algorithms are naturally resistant to mild forms of noise in the labels, such as ``classification noise'', where the label for each training example is flipped with some constant probability \cite{kearns1998efficient}.
	\item It is the most realistic learning model for which strong, unconditional lower bounds are known for many basic classes. Indeed, there is a considerable literature on this topic, with lower bounds usually proven using the so-called SQ-dimension and its generalizations \cite{blum1994weakly, feldman2012complete, reyzin2020statistical}.
	\item SQ algorithms are naturally implementable in a way that satisfies differential privacy of the training data, and indeed are the main examples of realistic differentially private learning algorithms \cite{blum2005practical, dwork2014algorithmic}.
\end{itemize}

Given all of these properties, it is natural to wonder whether the SQ model has something to bring to quantum learnability, with a particular eye towards noise tolerance. In this work we show (among other results) that for stabilizer states this approach cannot work: SQ-learning stabilizer states is exponentially hard, and in general, learning stabilizer states with noise is as hard as the well-known Learning Parity with Noise (LPN) problem.

\begin{theorem}
	Let $D$ denote the uniform distribution on Pauli measurements. Any SQ algorithm for learning $n$-qubit stabilizer states under $D$ up to error $2^{-O(n)}$ (i.e.\ to significantly outperform the maximally mixed state) requires $2^{\Omega(n^2)}$ queries even when tolerance is $2^{-O(n^2)}$.
\end{theorem}

Define a \emph{parity measurement} to be a Pauli measurement of the form $E_x = \frac{P_x + I}{2}$ for some $x \in \{0, 1\}^n$, where $P_x = \sum_{y \in \{0, 1\}^n} \chi_x(y) \ketbra{y}{y}$ and $\chi_x(y) = (-1)^{x \cdot y \bmod 2}$. They are so named since for any computational basis state $\ketbra{y}{y}$, $\trc(E_x \ketbra{y}{y}) = x \cdot y \bmod 2$. The following theorems hinge on the observation (stated as \cref{prop:parity_hardness}) that parities can be very naturally embedded within the problem of learning stabilizer states under distributions on parity measurements.
\begin{theorem}
	Let $D'$ denote the uniform distribution on parity measurements. Any SQ algorithm for learning $n$-qubit stabilizer states under $D'$ even up to constant error (say $1/3$) requires $2^{\Omega(n)}$ queries even when tolerance is $2^{-O(n)}$. 
\end{theorem}

\begin{theorem}
	Let $D'$ be as above. Learning $n$-qubit stabilizer states under $D'$ with classification noise at rate $\eta$ is at least as hard as the classical problem of Learning Parity with Noise (LPN) at rate $\eta$.
\end{theorem}

(These theorems are formally stated as \cref{cor:stab-maxmixed-sq,cor:stab-sq-lower,cor:stab-lpn} respectively.)

Our results position the problem of learning stabilizer states as a quantum analogue of the important classical problem of learning parities.\footnote{This recalls the way in which stabilizer codes are a quantum analogue of classical parity check codes.} In both cases there are simple ``algebraic'' learning algorithms for the noiseless setting, and the problem seems to become intractable with even the simplest kinds of noise. The algorithm of \citet{rocchetto2018stabiliser} thus joins a small class of PAC algorithms that do not fall into the SQ model, and hence do not admit any straightforward algorithms in noisy settings. In our view, this frames learning stabilizer states with noise as one of the more compelling problems on the frontier of learning quantum states with noise.

Another interpretation of our results is that they highlight limitations of the PAC framework of \citet{aaronson2007learnability}: insofar as this framework reduces the problem of learning quantum states to an essentially classical problem, it also inherits longstanding problems from classical learning theory. 

We also hope that our introduction of the SQ model to quantum state learning will be of independent interest and help spur new ideas in this area.

We now detail the rest of our contributions and lay out the organization of this paper:
 \begin{itemize}
	\item In \cref{sec:prelim}, we formally define the problem of SQ-learning quantum states and extend the notion of the SQ-dimension to this setting, building on recent work that formally analyzed the SQ-dimension as applicable to the p-concept setting \cite{goel2020superpolynomial}.
	\item In \cref{sec:sq-noise-tolerance}, we show that SQ algorithms for learning quantum states are indeed resistant to mild forms of noise, including classical classification noise as well as quantum channels with bounded noise (such as depolarizing noise).
	\item In \cref{sec:sq_lower}, we give exponential SQ lower bounds on learning stabilizer states. Under the uniform distribution on Pauli measurements, we show (Corollary \ref{cor:stab-maxmixed-sq}) that it requires exponentially many queries in order to improve on the maximally mixed state's performance. Under a different natural distribution on Pauli measurements, namely the uniform distribution over parity measurements, we show (Corollaries \ref{cor:stab-sq-lower}, \ref{cor:stab-lpn}) that learning stabilizer states with noise is as hard as learning parities with noise. 
	\item In \cref{sec:sq-learner}, by way of positive results, we give SQ algorithms for the simple setting of learning product states. We describe an SQ algorithm for learning product states under Haar-random single-qubit measurements, and show that it allows one to perform tomography on the individual qubits.
	\item In \cref{sec:diff-priv}, we relate SQ learning to a form of  differential privacy for quantum state learners. This form of differential privacy has recently been studied by \cite{arunachalam2021private}.
\end{itemize}

\subsection{Related work}
The problem of learning quantum states via state tomography has a long history in quantum computing, culminating in the celebrated optimal algorithms of O'Donnell and Wright \cite{o2016efficient, o2017efficient} and \citet{Haah_2017}. We operate in the alternative PAC framework introduced by \citet{aaronson2007learnability}. In recent years, this framework has been extended to the online setting \cite{aaronson2019online} as well as verified in experimental setups \cite{rocchetto2019experimental}. To our knowledge, the only known computationally efficient PAC learners for supervised learning of a commonly-considered class of states are the algorithm of \citet{rocchetto2018stabiliser} for learning stabilizer states, as well as that of \citet{yoganathan2019condition} for other classes of states whose generating circuits can be efficiently classically simulated and inverted, including low Schmidt rank states. While the focus of this paper is on stabilizer states, we remark that Yoganathan's algorithm for low Schmidt rank states also involves solving a system of polynomial equations in the examples, and hence would also not admit any straightforward SQ implementation. \citet{cheng2015learnability} frame the problem of PAC-learning unknown quantum measurements under a distribution of states as a dual problem to PAC-learning an unknown state, and are able to recover Aaronson's main sample complexity bound using a classical proof.

Other examples of efficient learning algorithms not lying strictly within the PAC formalism include \citet{aaronson2021efficient} in the unsupervised setting for non-interacting fermions and \citet{montanaro2017learning} for improved tomography of stabilizer states. \citet{lai2021learning} extend Montanaro's ideas for stabilizer states and use it to learn outputs of Clifford circuits with a single layer of $T$ gates.

Recent work by \citet{arunachalam2021private} extends work by \citet{bun2020equivalence} to the quantum setting, and relates differentially private (DP) learning of quantum states to one-way communication, online learning, and other models. We show in \cref{sec:diff-priv} that our notion of SQ learnability implies their notion of DP learnability, and hence by their results also implies finite sequential fat-shattering dimension, online learnability, and ``quantum stability.''

We note that the problem of PAC-learning quantum states is very different from the problem of PAC-learning Boolean functions using quantum representations of data, as considered in a recent active line of work \cite{arunachalam2018optimal, arunachalam2017guest}. In particular, the model of SQ-learning that we introduce is unrelated to a recent notion of SQ-learning of \emph{Boolean} functions using quantum representations \cite{arunachalam2020quantum}. When one is given quantum samples of Boolean or integer-valued functions, there have been important results on learning in the presence of noise, showing that both Learning Parity with Noise (LPN) \cite{cross2015quantum} and Learning With Errors (LWE) \cite{grilo2019learning} are tractable in this setting.

\section{Preliminaries}\label{sec:prelim}

\paragraph{Notation and terminology.} We use $\rho$ to refer to the density matrix of an $n$-qubit quantum mixed state, representable as a $2^n \times 2^n$ PSD operator of trace 1. (The number of qubits will be $n$ throughout and suppressed from the notation.) A pure state is a quantum state with rank $1$. Let $\E$ denote the space of two-outcome $n$-qubit measurements $E$ (corresponding to the POVM $\{E, I-E\}$), which accept a state $\rho$ with probability $\trc(E \rho)$. We will view the measurement outcomes themselves as $\{-1, 1\}$-valued, so that the outcome of measuring $\rho$ using $E$ is a random variable $Y$ that is $1$ with probability $\trc(E\rho)$ and $-1$ otherwise. Define $f_\rho : \E \to [-1, 1]$ to be the conditional mean function \[ f_\rho(E) = \ex[Y|E] = 1 \cdot \trc(E \rho) + (-1) \cdot (1 - \trc(E\rho)) = 2\trc(E\rho) - 1. \] We will often identify a state $\rho$ with its behavior with respect to two-outcome measurements, namely with the function $f_\rho$, and use the notation $Y \sim f_\rho(E)$ to mean that $Y \in \{-1, 1\}$ is the random measurement outcome satisfying $\ex[Y|E] = f_\rho(E)$. In learning theoretic terms, this means $f_\rho$ describes a \emph{probabilistic concept}, or p-concept, on the space $\E$. A p-concept on a domain $\X$ is a classification rule that assigns random $\{-1, 1\}$-valued labels to each point in $\X$ according to a fixed conditional mean function; we always identify the p-concept with its conditional mean function. Given a set $\F$ of quantum states, we use $\F$ to also mean the class of associated p-concepts, with the meaning clear from context.

Given a distribution $D$ over $\E$, we will often regard functions $f_\rho, f_\sigma : \E \to \R$ as members of the $L^2$ space $L^2(D, \E)$, with the inner product given by $\inn{f_\rho, f_\sigma}_D = \ex_{E \sim D}[f_\rho(E) f_\sigma(E)]$, and the norm given by $\|f_\rho\|_D = \sqrt{\inn{f_\rho, f_\rho}_D} = \sqrt{\ex_{E \sim D} [(f_\rho(E))^2]}$.

We use $[n]$ to refer to the set of indices $\{1, \dots, n\}$. Given a set $S \subseteq [n]$, we will use $\chi_S$ to refer to the parity on $S$, defined as a function from $\{-1, 1\}^n \to \{-1, 1\}$ by $\chi_S(x) = \prod_{i \in S} x_i$. Given $x, y \in \{0, 1\}^n$, we will sometimes also use $\chi_x(y) = (-1)^{x \cdot y \bmod 2}$.

\subsection{Learning models}
We begin by formally defining the problem of PAC-learning a quantum state.
\begin{definition}[PAC-learnability of quantum states, \cite{aaronson2007learnability}]
	Let $\F$ be a class of $n$-qubit quantum states. Let $D$ be a distribution over $\E$. We say $\F$ is PAC-learnable up to squared loss $\epsilon$ with respect to $D$ if there exists a learner that, given sample access to labeled examples $(E, Y)$ for $E \sim D, Y \sim f_\rho(E)$ for an unknown $\rho \in \F$, is able to output a state $\sigma$ satisfying \[ \ex_{E \sim D} \left[ (f_\sigma(E) - f_\rho(E))^2 \right] \leq \epsilon. \]
	The number of examples used by the learner is called its sample complexity.
\end{definition}
We note that this is a slight modification of the original definition in \cite{aaronson2007learnability}, stated directly in terms of squared loss since this is the view that will be convenient for us. PAC learners are also allowed to fail with some probability $\delta$, but for simplicity we will ignore this in this paper. A learner that only succeeds with some constant probability can easily be amplified to succeed with probability $1 - \delta$ using standard confidence amplification procedures.

With PAC learners, one may speak of both computational efficiency (overall running time) and statistical or information-theoretic efficiency (sample complexity). The original result of \citet{aaronson2007learnability} described a computationally inefficient algorithm for learning arbitrary states that nevertheless had $O(n)$ sample complexity. An \emph{efficient} PAC learner is one that is computationally efficient, i.e.\ runs in polynomial time, and hence also draws at most polynomially many examples (each draw is considered as taking one unit of time).

We now introduce the following natural extension of these definitions to the SQ setting. In both cases, we operate in the so-called distribution-specific setting, where the learner is assumed to have knowledge of the distribution $D$.

\begin{definition}[SQ-learnability of quantum states]
	Let $\F$ be a class of $n$-qubit quantum states. Let $D$ be a distribution over $\E$. An SQ oracle for an unknown state $\rho \in \F$ is an oracle that accepts a query and a tolerance, $(\phi, \tau)$, where $\phi : \E \times \{-1, 1\} \to [-1, 1]$ and $\tau > 0$, and responds with $y$ such that \[ \Big\lvert y - \ex_{E \sim D, Y \sim f_\rho(E)} \left[ \phi(E, Y) \right] \Big\rvert \leq \tau. \] We say $\F$ is SQ-learnable up to squared loss $\epsilon$ if there is a learner that, given only queries to the SQ oracle for an unknown $\rho \in \F$, is able to output a state $\sigma$ satisfying \[ \ex_{E \sim D} \left[ (f_\sigma(E) - f_\rho(E))^2 \right] \leq \epsilon. \] The number of queries used by the learner is called its query complexity.
\end{definition}

An SQ learner is considered \emph{efficient} if it uses polynomially many queries and its queries all have tolerance $\tau \geq 1/\poly(n)$.

\subsection{SQ lower bounds for p-concepts}
One of the chief features of the classical SQ model is the possibility of proving unconditional lower bounds on learning a class $\cC$ in terms of its so-called statistical dimension. The quantum setting that we work in, where we identify a state $\rho$ with the p-concept $f_\rho$, becomes a special case of the SQ model for learning p-concepts. Building on recent work \cite{goel2020superpolynomial} that formally proved SQ lower bounds for p-concepts, we extend this framework to the quantum setting. Let $\X$ denote an arbitrary domain (for us, $\X$ will be $\E$, while in the classical setting, $\X$ is usually $\R^n$). 

\begin{definition}[Statistical dimension]
	Let $D$ be a distribution on $\X$, and let $\cC$ be a class of functions from $\X$ to $\R$. The  \emph{average (un-normalized) correlation} of $\cC$ is defined to be $ \rho_D(\cC) = \frac{1}{|\cC|^2} \sum_{c, c' \in \cC} |\inn{c, c'}_D|$.
	The \emph{statistical dimension on average} at threshold $\gamma$, $\sda_D(\cC, \gamma)$, is the largest $d$ such that for all $\cC' \subseteq \cC$ with $|\cC'| \geq |\cC|/d$, $\rho_D(\cC') \leq \gamma$.
\end{definition}

\begin{theorem}[\cite{goel2020superpolynomial}, Cor.\ 4.6]\label{thm:sq_lower_bound}
	Let $D$ be a distribution on $\X$, and let $\cC$ be a p-concept class on $\X$. Say our queries are of tolerance $\tau$, the final desired squared loss is $\epsilon$, and that the functions in $\cC$ satisfy $\|c\|_D \geq \beta$ for all $c \in \cC$. For technical reasons, we require $\tau \leq \epsilon$,  $\epsilon^2 \leq \beta/3$. Then learning $\cC$ up to squared loss $\epsilon$ (we may pick $\epsilon$ as large as $\sqrt{\beta/3}$) requires at least $\sda_D(\cC, \tau^2)$ queries of tolerance $\tau$.
\end{theorem}

We remark that the way to interpret such a lower bound is as follows: if the SQ learner's queries have tolerance at least $\tau$, then at least $\sda_D(\cC, \tau^2)$ queries are required. That is, one must \emph{either} use small tolerance \emph{or} many queries.

The following lemma will be convenient in order to bound the $\sda$ when we have bounds on pairwise correlations. 

\begin{lemma}[\cite{goel2020superpolynomial}, Lemma 2.6]\label{lemma:sda}
	Let $D$ be a distribution on $\X$, and let $\cC$ be a p-concept class on $\X$ such that for all $c, c' \in \cC$ with $c \neq c'$, $|\inn{c, c'}_D| \leq \gamma$, and for all $c \in \cC$, $\|c\|_D^2 \leq \kappa$. Then for any $\gamma' > 0$, $\sda(\cC, \gamma + \gamma') \geq |\cC| \frac{\gamma'}{\kappa - \gamma}$.
\end{lemma}

\subsection{The problem of learning parities}
One of the most basic problems in classical learning theory is that of learning the concept class of parity functions. Let the domain be $\X = \{-1, 1\}^n$, and for any subset $S \subseteq [n]$, define $\chi_S(x) = \oplus_{i \in S} x_i = \prod_{i \in S} x_i$ to be the parity on $S$. Here $x_i \oplus x_j = x_i x_j$ is simply the XOR when bits are represented by $\{-1, 1\}$. Let $D$ be any distribution on $\X$. We say a learner is able to learn parities under $D$ if given access to labeled examples $(x, \chi_S(x))$ where $x \sim D$ and $S \subseteq [n]$ is unknown (or, in the SQ setting, given access to the corresponding SQ oracle), and for any error parameter $\epsilon$, it is able to output a function $h$ such that $\pr_{x \sim D}[h(x) \neq \chi_S(x)] \leq \epsilon$.

The problem of learning parities displays a striking phase transition in going from the noiseless to the noisy setting. Given noiseless labeled examples, the problem of recovering the right parity is simply a question of solving linear equations over $\mathbb{F}_2$, and can be done using Gaussian elimination by a PAC learner using only $\Theta(n)$ examples. With just a little noise, however, the problem seems to become intractable. Perhaps the simplest noise model one can consider is the classification noise model, where every example has its label flipped with some constant probability $\eta$ (known as the noise rate). Learning parities under classification noise is the basis of the famous \emph{Learning Parity with Noise (LPN)} problem. Formally, the search version of LPN with noise rate $\eta$ is precisely the problem of learning parities under the uniform distribution on $\X$ and with classification noise at rate $\eta$. Usually one also has the additional knowledge that the true target $\chi_S$ (the ``secret'') is picked uniformly at random from the set of all parities. This problem is widely conjectured to be hard, including for quantum algorithms, and is even used as a basis for cryptography (see \cite{pietrzak2012cryptography} for a survey). The best-known algorithm in the PAC setting runs in slightly subexponential time \cite{blum2003noise}.

Since SQ learners are naturally tolerant of classification noise, one would expect that there are no SQ learners for parities under the uniform distribution, and indeed, this is one of the foundational results in the SQ literature.
\begin{theorem}[\cite{kearns1998efficient}]\label{thm:kearns_parity}
	Any SQ learner requires $2^{\Omega(n)}$ queries (even using tolerance $2^{-O(n)}$) to learn parities under the uniform distribution on $\{-1, 1\}^n$ even up to constant error (say $1/3$).
\end{theorem}

Thus we see that simple Gaussian elimination is an example of an efficient PAC learner that is not SQ. This establishes a characteristic limitation of SQ algorithms: while they include a wide range of common algorithms, they do not include algorithms that depend entirely on ``algebraic'' structure.

It is worth emphasizing that this discussion has considered learning parities with a \emph{classical} representation of the data. When given a \emph{quantum} representation of the data, as in the quantum ``example state'' $\ket{\psi} = 2^{-n/2} \sum_{x \in \{0, 1\}^n} \ket{x, \chi_S(x)}$ (taking the distribution over the domain to be uniform), the task becomes easy even with noise \cite{cross2015quantum}. This is because we can now use Hadamard gates to implement a Boolean Fourier transform \`{a} la the famous Bernstein--Vazirani algorithm \cite{bernstein1997quantum}. 

\section{Noise-tolerant SQ learning}\label{sec:sq-noise-tolerance}
One of the prime features of classical SQ learning is its inherent noise tolerance. From an intuitive standpoint, certain common stochastic noise models are systematic enough that their effects \emph{in expectation} can be predicted in advance, and hence either be corrected for or bounded. Slightly more precisely, the query expectations of a noisy state are often related in simple ways to the query expectations on a noiseless state, so that the latter can be recovered from the former. We mainly consider three such noise models here: (a) classical classification noise and malicious noise, (b) quantum depolarizing noise, and (c) more general quantum channels with bounded noise.

\subsection{Classification and malicious noise}
Classification noise \cite{angluin1988learning} and malicious noise \cite{valiant1985learning,kearns1993learning} are two classical Boolean noise models that SQ algorithms are able to handle. In the classification noise model, every example's label is flipped with probability $\eta$ (known as the noise rate). The malicious noise model is a stronger form of noise where for any given example, with probability $1 - \eta$, the label is reported correctly, but with probability $\eta$ both the point and its label may be arbitrary (and adversarially selected based on the learner's behavior so far). We note that these models are well-defined even in the p-concept setting and hence for quantum states, and simply introduce further randomness into the label. The following results were originally stated for Boolean functions but readily extend to p-concepts.

\begin{theorem}[\cite{kearns1998efficient}]
	Let $\cC$ be a p-concept class learnable under distribution $D$ in the SQ model up to error $\epsilon$ using $q$ queries of tolerance $\tau$. Then for any constant $0 < \eta < 1/2$, even with respect to an SQ oracle with classification noise at rate $\eta$ (i.e., one that computes expectations with classification noise), $\cC$ is learnable up to $\epsilon$ using $O(q)$ queries of tolerance $O(\tau (1 - 2\eta))$. If the learner is given noisy training examples as opposed to access to a noisy SQ oracle, then $\tilde{O}(\frac{q}{\poly(\tau(1 - 2\eta))})$ noisy examples suffice.
\end{theorem}

\begin{theorem}[\cite{aslam1998specification}]
	Let $\cC$ be a p-concept class learnable under distribution $D$ in the SQ model up to error $\epsilon$ using $q$ queries of tolerance $\tau$. An SQ oracle with malicious noise at rate $\eta$ is one that computes query expectations with respect to a distribution $(1 - \eta)f(D) + \eta Q$, where $f(D)$ denotes the true labeled distribution $(x, y)$ for $x \sim D, y \sim f(x)$ ($f$ being the unknown target p-concept), and $Q$ is an arbitrary and adversarially selected distribution on $\X \times \{-1, 1\}$. If $\eta = \tilde{O}(\epsilon)$ and $\eta < \tau$, then even with respect to an SQ oracle with malicious noise at rate $\eta$, $\cC$ is learnable up to $\epsilon$ using $O(q)$ queries of tolerance $\tau - \eta$. If the learner is given noisy training examples as opposed to access to a noisy SQ oracle, then $\cC$ is learnable (with constant probability) using $\tilde{O}(\frac{q}{\poly(\tau - \eta)})$ noisy examples suffice. (More efficient implementations are also available in some special cases).
\end{theorem}

The proofs of both theorems are similar: one first relates the noisy query expectations to the true expectations, and then argues that when using a suitably small tolerance (or sufficiently many examples) the effects of the noise can be corrected for (within information theoretic limits). 

\subsection{Depolarizing noise}
Depolarizing noise acts on quantum states by shifting them closer to the maximally mixed state. One can consider a setting where it acts on an entire $n$-qubit state at once, as well as one where it acts independently on each individual qubit. We will consider the former.

\begin{definition}[Depolarizing noise]
	Let $\rho$ be an arbitrary $n$-qubit state. Then depolarizing noise at rate $\eta$ ($0 < \eta < 1$) acts on this state by transforming it into $\Lambda_{\eta}(\rho) = (1 - \eta)\rho + \eta (I/2^n)$.
\end{definition}

\begin{theorem}
	Let $0 < \eta < 1$ be any constant, and let $\Lambda_\eta$ denote the depolarizing channel at noise rate $\eta$. Let $\cC$ be a class of $n$-qubit quantum states and $D$ be a distribution on $\E$, the space of two-outcome measurements on such states. Let $L$ be an SQ learner capable of learning $\cC$ under $D$ using $q$ queries of tolerance $\tau$. There there exists a learner $L'$ such that for any $\rho \in \cC$, $L'$ is capable of learning $\rho$ under $D$ using $q$ queries of tolerance $\tau(1 - \eta)$ given only SQ access to $\Lambda_\eta(\rho)$ as well as sampling access to $D$.
\end{theorem}
\begin{proof}
	For simplicity, we will assume that we know the noise rate $\eta$ exactly. (So long as we have an upper bound on $\eta$, then by a standard ``grid search'' argument due to \cite{kearns1998efficient}, we can estimate $\eta$ sufficiently closely simply by trying out many different values. Briefly: if say we try out $\eta = 0, \delta, 2\delta, \dots, 1$ ($1/\delta$ values in all), then one of these will be within $\delta/2$ of the true $\eta$. The algorithm when run with this guess for $\eta$ will produce a good hypothesis. By taking $\delta = O(\tau(1 - \eta)^2)$ and testing all $1/\delta$ hypotheses produced by our guesses for $\eta$ on a sufficiently large validation set, we can ensure the best one will perform and generalize well.)
	
	Let $\rho \in \cC$ be the unknown target. Observe that for any $E \in \E$, by linearity, \[ f_{\Lambda_\eta(\rho)}(E) = 2\trc(E \cdot \Lambda_\eta(\rho)) - 1 = (1 - \eta)f_\rho(E) + \eta f_{I/2^n}(E). \] Let $\phi : \E \times \{-1, 1\} \to [-1, 1]$ be any query made by $L$. Let $\phi[\rho]$ denote the query expectation of $\phi$ under $\rho$, given by $\ex_{x \sim D} \ex_{y \sim f_\rho(x)}[\phi(x, y)]$. Similarly let the noisy analogue be $\phi[\Lambda_{\eta}(\rho)]$. Again just by linearity, \[ \phi[\Lambda_{\eta}(\rho)] = (1 - \eta) \phi[\rho] + \eta \phi[I/2^n]. \] The latter quantity is independent of $\rho$ and can be estimated to arbitrary accuracy by sampling from $D$, allowing us to estimate $\phi[\rho]$ as $\frac{\phi[\Lambda_{\eta}(\rho)] - \eta \phi[I/2^n]}{1 - \eta}$. So long as $\eta$ is bounded away from $1$, we can use a query of tolerance $\tau (1 - \eta)$ to estimate $\phi[D_{\eta}(\rho)]$ (as well as $1/\poly(\tau(1-\eta))$ unlabeled examples from $D$ to compute $\phi[I/2^n]$), and thereby estimate $\phi[\rho]$ to within $\tau$. Thus we can simulate $L$ even with depolarizing noise.
\end{proof}

It is worth stressing that we are able to handle any constant noise rate $\eta \in (0, 1)$, and the price we pay is requiring the tolerance to scale as $\tau(1 - \eta)$.

\subsection{Quantum channels with bounded noise}

We can also consider more general kinds of quantum channels with bounded noise. As long as the queries are bounded, small amounts of noise cannot alter query expectations too much, and so can be ``absorbed'' into the tolerance. This is similar to classical malicious noise: since classical malicious noise at rate $\eta$ only can only change query expectations by $\eta$ (recall that the queries are bounded by 1), a noisy query of tolerance $\tau - \eta$ is able to simulate a noiseless query of tolerance $\tau$. Unlike with depolarizing noise, this means we cannot handle arbitrary $\eta$; this is an artifact of the fact that more general kinds of noise do not permit the kind of systematic correction we were able to perform for depolarizing noise.

For concreteness here we consider a noisy quantum channel $\Lambda$ such that $\lVert \Lambda - 1_n \rVert_\diamond \leq \eta$, where $1_n$ is the identity map on $n$-qubit states and the norm is the diamond norm. We do not define this norm here, but its chief property for our purposes is that for any $n$-qubit state $\rho$ and 2-outcome measurement $E$, $\lvert \trc(E(\rho - \Lambda(\rho))) \rvert \leq \eta$. Similar theorems can be proven with respect to other distance measures such as fidelity.

\begin{theorem}
	Let $\Lambda$ be a quantum channel such that $\lVert \Lambda - 1_n \rVert_\diamond \leq \eta$, as above.  Let $\mathcal{C}$ be a class of $n$-qubit quantum states learnable under distribution $D$ using $q$ queries of tolerance $\tau > 2\eta$. Then $\mathcal{C}$ is still learnable under noise $\Lambda$ (i.e.\ when our queries are answered not with respect to $\rho$ but $\Lambda(\rho)$) using $q$ noisy queries of tolerance $\tau - 2\eta$.
\end{theorem}
\begin{proof}	
	As noted, for any state $\rho$ and measurement $E$, $\lvert \trc(E(\rho - \Lambda(\rho))) \rvert \leq \eta$. Consider any query $\phi: \E \times \{-1, 1\} \to [-1, 1]$. If $\phi[\rho]$ denotes the query expectation on a noiseless state and $\phi[\Lambda(\rho)]$ denotes the noisy one, then a straightforward calculation shows that \begin{align*}
		\left| \phi[\rho] - \phi[\Lambda(\rho)] \right|
		&= \Big\lvert \ex_{E \sim D} \big[(\phi(E, 1) - \phi(E, -1))\trc(E (\rho - \Lambda(\rho))) \big] \Big\rvert \\
		& \leq 2 \ex_{E \sim D} \lvert\trc(E (\rho - \Lambda(\rho)))\rvert\\
		& \leq 2 \eta.
	\end{align*} Thus just by the triangle inequality, if we calculate $\phi[\Lambda(\rho)]$ within tolerance $\tau - 2\eta$, then we also get $\phi[\rho]$ within $\tau$.

\end{proof}

\subsection{General noise for distribution-free learning}

So far, we've only considered distribution-specific learning, where the learner is only required to succeed with respect to a pre-specified distribution $D$. In the distribution-free case, where the learner is required to succeed no matter what $D$ is, we now give a simple proof that any SQ algorithm for a concept class can also handle any kind of quantum noise on the state, as long as the noise is known. This is unsurprising, and at a high level, the approach simply boils down to off-loading the noise from the state to the measurement. Learning a noisy set of measurements is thus handled by distribution-free learning algorithm.

Given a quantum operation $\Lambda$, its adjoint $\Lambda^\dagger$ is such that $\forall \rho, \trc(E \cdot \Lambda(\rho)) = \trc(\Lambda^\dagger(E) \cdot \rho)$ and always exists (see \cite{Rall_2019} for details on how to prove this folklore result). Let $\mathcal{D}$ be the distribution we are trying to learn concept class $\cC$ using statistical queries and let $\Lambda$ be the noise applied to the quantum state. We can then define $\mathcal{D^\dagger}$ to be the distribution $\Lambda^\dagger(E)$ where $E$ is drawn from $\mathcal{D}$ and by definition the traces (and thus the statistical queries) are the same when applied to $\rho$ and $\Lambda(\rho)$ respectively. Also by definition, a distribution-free learner for $\cC$ would also be able to learn with distribution $\mathcal{D^\dagger}$.

\section{Lower bounds on learning stabilizer states with noise} \label{sec:sq_lower}

\subsection{Stabilizer formalism}
The stabilizer states are a popular class of quantum states that are used throughout quantum information for areas like quantum error correction, classical simulation of quantum mechanics, and quantum communication. If we let $\mathcal{P} = \{I, X, Y, Z\}$ be the Pauli matrices, then we can define a generalization of the Pauli matrices to an $n$-qubit system as $\mathcal{P}_n = \{\pm 1\} \cdot \{I, X, Y, Z\}^{\otimes n}$.\footnote{Note that this is slightly different from the usual Pauli group $ \{\pm 1, \pm i\} \cdot \{I, X, Y, Z\}^{\otimes n}$. We choose to ignore the imaginary phases as they will never be a part of a valid stabilizer measurement: for instance, $(I + iX)/2$ is not even Hermitian.} As an example, $-X \otimes Y \otimes Z \otimes Y \otimes Z \in \mathcal{P}_5$, though in the future we will simply opt to write this Pauli matrix as $-XYZYZ$. It is not hard to show that $\forall P \neq \pm I^{\otimes n} \in \mathcal{P}_n$ half of the eigenvalues are $1$ and the other half are $-1$.

We say that a pure state $\rho = \ket{\psi}\bra{\psi}$ such that $\ket{\psi} \in \mathbb{C}^{2^n}$ is \emph{stabilized} by $P \in \mathcal{P}_n$ if $P \ket{\psi} = \ket{\psi}$. In other words, $\ket{\psi}$ must be an eigenvector of $P$ with eigenvalue 1. The set of pure states that are stabilized by the subset $S \subseteq \mathcal{P}_n$ are then the states that lie in the intersection of the eigenvalue 1 subspaces. It is known that if $S$ is an Abelian group that does not contain $-I^{\otimes n}$, then this intersection has dimension $2^n/\lvert S \rvert$ \cite{hamermesh1989group}. Due to the nature of the outer product, any vector $e^{i \theta}\ket{\psi}$ drawn from this space results in the same density matrix $\rho = \ketbra{\psi}{\psi}$.

\begin{definition}[Stabilizer states]
	Let $S \subset \mathcal{P}_n \setminus \{-I^{\otimes n}\}$ be an Abelian group of order $2^n$ (note that $\mathcal{P}_n$ is not itself a group under matrix multiplication). The unique state density matrix $\rho = \ketbra{\psi}{\psi}$ that results from the one-dimensional subspace that is stabilized by all elements of $S$ is then defined to be a stabilizer state. We then say that $S$ stabilizes $\rho$.
\end{definition}

The stabilizer states are then the set of all such quantum pure states that are stabilized by Abelian groups of order $2^n$ formed from $\mathcal{P}_n \setminus \{-I^{\otimes n}\}$.

\begin{proposition}\label{prop:stabilizer_intersection}
	Given two $n$-qubit stabilizer states $\ketbra{\psi}{\psi} \neq \ketbra{\phi}{\phi}$ with stabilizer groups $S$ and $S'$ respectively, then $\lvert S \cap S'| \leq 2^{n-1}$.
\end{proposition}
\begin{proof}
	Because $\ketbra{\psi}{\psi} \neq \ketbra{\phi}{\phi}$ then $S \neq S'$. We also know that $S \cap S'$ is an abelian group without $-I^{\otimes n}$, so $|S \cap S'| < 2^n$. Since $2^n / |S \cap S'|$ is the dimension of the space stabilized by this group \cite{hamermesh1989group}, it must be an integer. Due to the prime factorization of $2^n$, $|S \cap S'| = 2^m$ for some integer $0 \leq m < n$, of which the largest possible $m$ is $n-1$.
\end{proof}

\subsection{Difficulty of beating the maximally mixed state on uniform Pauli measurements}

Let $\cC$ be the class of all $n$-qubit stabilizer pure states. If $P \in \mathcal{P}_n$ is a Pauli operator, then the two-outcome measurement associated with $P$ is $(I + P)/2$, and is referred to as a Pauli measurement. We will first examine the natural distribution $D$ given by the uniform distribution over Pauli measurements, \[ \Epauli = \{ (I + P)/2 \mid P \in \mathcal{P}_n \}. \] In doing so, we will show that performing better than the trivial algorithm of always outputting the maximally mixed state $I/2^n$ is difficult.

We use the following folklore lemma (see e.g. \cite{rocchetto2018stabiliser} Lemma 1 for a simple proof).

\begin{lemma}\label{lemma:pauli_measure_trace}
	Let $E = (I + P)/2$ be a POVM measurement associated to a Pauli operator $P \in \mathcal{P}_n$ and $\rho$ be an $n$-qubit stabiliser state. Then $\trc(E\rho)$ can only take on the values $\{0, \frac{1}{2}, 1\}$, and:
	\[
	\begin{cases}
		\trc(E \rho) = 1 \text{ iff $P$ is a stabilizer of $\rho$;}\\
		\trc(E \rho) = 1/2 \text{ iff neither $P$ nor $-P$ is a stabilizer of $\rho$;}\\
		\trc(E \rho) = 0 \text{ iff $-P$ is a stabilizer of $\rho$.}
	\end{cases}
	\]
\end{lemma}

Simple algebraic manipulations then tell us that $\trc(P \rho)$ can only take on the values $\{1, 0, -1\}$ with $\trc(P \rho)$ being $1$ or $-1$ if and only if $P$ or $-P$ is in the stabilizer group of $\rho$ respectively. With this result, we can compute bounds on $|\inn{f_\rho, f_{\rho'}}_D|$ for stabilizer states $\rho$ and $\rho'$ by counting how many matrices lie in the intersection of their stabilizer groups or the negations of their stabilizer groups.

\begin{lemma}\label{lemma:stabilizer_correlations}
	Let $\cC$ be the concept class of $n$-qubit stabilizer pure states, and let $D$ denote the uniform distribution on $n$-qubit Pauli measurements. Then for any stabilizer states $\rho, \rho'$ with $\rho \neq \rho'$, $|\inn{f_\rho, f_{\rho}}_D| = \|f_\rho\|_D^2 = \frac{1}{2^n}$, and $|\inn{f_\rho, f_{\rho'}}_D| \leq \frac{1}{2^{n+1}}$. Furthermore, this inequality is tight.
\end{lemma}
\begin{proof}
	Let $\rho, \rho' \in \cC$. Let $S$ and $S'$ be the stabilizer groups for $\rho$ and $\rho'$ respectively, and also let $-S = \{-P \mid P \in S\}$ and $-S' = \{-P \mid P \in S'\}$. For any Pauli measurement $(I + P)/2$, we have $f_\rho(\frac{I+P}{2}) = 2\trc(\frac{I+P}{2}\rho) - 1 = \trc(P \rho)$. Thus the correlation of the two p-concepts becomes \[ |\inn{f_\rho, f_{\rho'}}_D| = \frac{1}{|\mathcal{P}_n|} \left| \sum_{P \in \mathcal{P}_n} \trc(P\rho)\cdot \trc(P \rho') \right|. \] Leveraging \cref{lemma:pauli_measure_trace}, we find:
	\begin{align*}
		|\inn{f_\rho, f_{\rho'}}_D| &= \frac{1}{|\mathcal{P}_n|} \left| \sum_{P \in \mathcal{P}_n} \trc(P\rho)\cdot \trc(P \rho') \right| \\
		&= \frac{1}{2 \cdot 4^n} \left[ |S \cap S'| + |-S \cap -S'| - |S \cap -S'| - |-S \cap S'| \right] \\
		&= \frac{1}{4^n} \left[ |S \cap S'| - |S \cap (-S')| \right]
	\end{align*}
	
	If $\rho = \rho'$ then $S = S'$ such that $|S \cap S'| = 2^n$ and $|S \cap -S'| = 0$. Thus for all $\rho$
	\[
	|\inn{f_\rho, f_{\rho}}_D|  = \frac{2^n}{4^n} = \frac{1}{2^n}
	\]
	
	If $\rho \neq \rho'$ then by \cref{prop:stabilizer_intersection} $|S \cap S'| \leq 2^{n-1}$. $|S \cap -S'| \geq 0$ trivially, so we get an upper bound of
	\[
	|\inn{f_\rho, f_{\rho'}}_D|  \leq \frac{1}{4^n}(2^{n-1}) = \frac{1}{2^{n+1}}.
	\]
	
	We can show that this inequality is tight because the state $\rho = \ketbra{0}{0}^{\otimes n}$ and $\rho' = \ketbra{0}{0}^{\otimes n-1} \otimes \ketbra{+}{+}$ saturate this inequality. The generators for the stabilizer group of $\rho$ are (omitting tensor products): $ZIII\cdots I$, $IZIII\cdots I$, ... , and $IIII\cdots IZ$. The generator of $\rho'$ are the same, except the last generator is replaced with $IIII\cdots IX$. We see that $|S\cap S'| = 2^{n-1}$ while $|S \cap -S'| = 0$.
	
\end{proof}

With this result, we can use \cref{lemma:sda} to compute the $\sda$ and by extension prove a lower bound on the number of statistical queries needed to learn this concept class under this distribution.

\begin{proposition}[\cite{aaronson2004improved} Proposition 2]\label{prop:num_stabilizer}
	The number of $n$-qubit stabilizer states grows as $2^{\Theta(n^2)}$.
\end{proposition}

\begin{theorem}
	Let $D$ be the uniform distribution over $n$-qubit Pauli measurements and let $\cC$ be the concept class of all $n$-qubit stabilizer pure states. Then $\sda(\cC, \frac{1}{2^{n}}) = 2^{\Theta(n^2)}$.
\end{theorem}
\begin{proof}
	By \cref{prop:num_stabilizer}, $|\cC| = 2^{\Theta(n^2)}$.
	Using \cref{lemma:sda} with $\kappa = \frac{1}{2^n}$ and $\gamma = \frac{1}{2^{n+1}}$ as calculated from \cref{lemma:stabilizer_correlations}, we find that 
	\[
	\sda(\cC, \gamma' + \frac{1}{2^{n+1}}) = \sda(\cC, \gamma' + \gamma) \geq |\cC|\frac{\gamma'}{\beta - \gamma} =  2^{\Theta(n^2)}\gamma' 2^{n+1} = 2^{\Theta(n^2)}\gamma'
	\] Setting $\gamma' = \frac{1}{2^{n+1}}$ gives the result.
\end{proof}

\begin{corollary}\label{cor:stab-maxmixed-sq}
	Any SQ algorithm needs at least $2^{\Omega(n^2)}$ statistical queries of tolerance $\tau = 2^{-O(n)}$ to learn $\cC$ up to error $2^{-O(n)}$ over $D$.
\end{corollary}
\begin{proof}
	Simply apply \cref{thm:sq_lower_bound}, with $\beta = 2^{-n}$.
\end{proof}

Since the norms of our p-concepts are exponentially small (i.e. $2^{-n/2}$), we only get hardness for error on the order of $2^{-O(n)}$. But as we now show, the p-concept norm corresponds almost exactly to the squared loss achieved by the maximally mixed state. Our results show that doing significantly better than the maximally mixed state requires $2^{\Omega(n^2)}$ statistical queries even when the tolerance is exponentially small.

\begin{proposition}\label{prop:maximally_mixed_error}
	Let $D$ be the uniform distribution over $n$-qubit Pauli measurements, $\Epauli$. Let $\rho$ be any state, and let $I/2^n$ be the maximally mixed state. Then \[ \|f_\rho\|_D^2 = \|f_\rho - f_{I/2^n}\|_D^2 + 1/4^n. \]
\end{proposition}
\begin{proof}
	In essence, this is simply because the p-concept $f_{I/2^n}$ is almost always zero. Specifically, for all $E \in \Epauli \setminus \{0, I\}$, $f_{I/2^n}(E) = 2\trc(E/2^n) - 1 = 0$, since $\trc(E) = 2^{n-1}$ for all such $E$. (This is because $E = (I + P)/2$ for some $n$-qubit Pauli matrix $P \in \cP_n$, and $\trc(P) = 0$ for all $P \in \cP_n \setminus \{\pm I\}$.) As for $E \in \{0, I\}$, we note that $f_\rho(E) = f_{I/2^n}(E)$. Thus
	
	\begin{align*}
		\|f_\rho\|_D^2 &= \frac{1}{|\Epauli|} \sum_{E \in \Epauli} f_\rho(E)^2 \\
		&= \frac{1}{|\Epauli|} \left( \sum_{E \in \Epauli \setminus \{0, I\}} f_\rho(E)^2 + \sum_{E \in \{0, I\}} f_\rho(E)^2 \right) \\
		&= \frac{1}{|\Epauli|} \left( \sum_{E \in \Epauli \setminus \{0, I\}} (f_\rho(E) - f_{I/2^n}(E))^2 + \sum_{E \in \{0, I\}} f_\rho(E)^2 \right) \\
		&= \frac{1}{|\Epauli|} \left( \sum_{E \in \Epauli} (f_\rho(E) - f_{I/2^n}(E))^2 + \sum_{E \in \{0, I\}} f_\rho(E)^2 \right) \\
		&= \|f_\rho - f_{I/2^n}\|_D^2 + \frac{2}{|\Epauli|} \\
		&= \|f_\rho - f_{I/2^n}\|_D^2 + 1/4^n.
	\end{align*}
\end{proof}

\subsection{Lower bounds via a direct reduction from learning parities}
To get around this norm issue, we look at a subset of stabilizer states such that we can produce p-concepts with norm 1. Recall that the Pauli measurements are the set of all projectors onto the eigenvalue-1 space of some Pauli matrix $P$, i.e.\ $\{\frac{P + I}{2} \mid P \in \mathcal{P}_n\}$. We define a subset of the Pauli measurements called the parity measurements, and show the hardness of SQ-learning stabilizer states under the uniform distribution on such measurements. This is via a simple equivalence, holding essentially by construction, with the problem of learning parities under the uniform distribution. As a further consequence, we obtain that learning stabilizer states with noise is at least as hard as Learning Parity with Noise (LPN). This holds for general PAC-learning, even outside the SQ model.

\begin{definition}[Parity measurements]
	For all $x \in \{0, 1\}^n$, let $P_x = \sum_{y \in \{0, 1\}^n} \chi_x(y) \ketbra{y}{y}$. Since the set of $P_x$ is equivalent to $\{I, Z\}^{\otimes n}$, the corresponding measurement $E_x = \frac{P_x + I}{2}$ is by definition a Pauli measurement. We will refer to such measurements as \textit{parity measurements}.
\end{definition}

\begin{proposition}\label{prop:parity_hardness}
For every distribution $D$ on $\{0,1\}^n$ there is a corresponding distribution $D'$ on parity measurements such that learning computational basis states under $D'$ is equivalent to learning parities under $D$. Furthermore, this equivalence holds even with classification noise: for any $\eta$, learning computational basis states under $D'$ with noise rate $\eta$ is equivalent to learning parities under $D$ with noise rate $\eta$.

In particular, learning stabilizer states under $D'$ is at least as hard as learning parities under $D$.
\end{proposition}
\begin{proof}
	If the unknown state $\rho$ is a computational basis state $\ketbra{y}{y}$, then the value
	\[
	\trc(E_x \ketbra{y}{y}) = x \cdot y \mod 2
	\]
	is simply the parity of $x$ over the subset specified by $y$ (represented using $\{0, 1\}$ instead of $\{-1, 1\}$). In the PAC setting, this would be equivalent to getting the sample $(E_x, x\cdot y \mod 2)$. Accordingly, let us define $D'$ simply as the distribution over $E_x$ for $x \sim D$. It is clear that these are different representations of the same problem, such that a learning algorithm for one implies a learning algorithm for the other. We note that this relationship holds even in the presence of classification noise. Finally, note that computational basis states are a subset of the stabilizer states, so any learner for stabilizer states implies a learner for the computational basis states as well. This implies that learning stabilizer states on $D'$ is at least as hard as learning parities on $D$, even in the presence of classification noise.
\end{proof}

\begin{corollary}\label{cor:stab-sq-lower}
	SQ-learning stabilizer states under the uniform distribution over parity measurements requires $2^{\Omega(n)}$ queries even with constant error (say $1/3$). 
\end{corollary}
\begin{proof}
	By \cref{prop:parity_hardness}, SQ-learning stabilizer states under the uniform distribution on $E_x$ parity measurements is at least as hard as learning parities over the uniform distribution. Applying \cref{thm:kearns_parity}, we get the exponential lower bound.
\end{proof}

\begin{corollary}\label{cor:stab-lpn}
	Learning stabilizer states under the uniform distribution over parity measurements with classification noise rate $\eta$ is at least as hard as LPN with noise rate $\eta$.
\end{corollary}
\begin{proof}
	\cref{prop:parity_hardness} directly implies that learning computational basis states under the uniform distribution on parity measurements and with classification noise is equivalent to LPN.
\end{proof}

\section{An SQ learner for product states}\label{sec:sq-learner}
Turning to positive results, we now give SQ algorithms for some simple concept classes, namely the computational basis states and, more generally, products of $n$ single-qubit states. The distribution on measurements that we will consider will correspond to a natural scheme for these classes: pick a qubit at random and measure it using a Haar-random unitary.

Concretely, let $D'$ be the distribution of single qubit measurements formed from the projection onto Haar-random single qubit state (i.e. $U\ketbra{0}{0}U^\dagger$ where $U$ is a Haar random unitary), and let $D$ be the distribution on $n$-qubit measurements that corresponds to picking a qubit at random and measuring it using a measurement drawn from $D'$. That is, $D = \frac{1}{n} \sum_{i=1}^n I^{\otimes i - 1} \otimes D' \otimes I^{\otimes n-i}$. Let $\cC$ be the concept class of product states $\rho = \otimes_{i=1}^n \rho_i$. Of course, this class includes the computational basis states. The main result of this section will be a simple $O(n)$-query SQ algorithm for learning $\cC$ under the distribution $D$.

We remark that our algorithm's guarantee actually trivially extends to learning arbitrary (not just product) states under the above distribution $D$ of single-qubit Haar-random measurements. This is simply because such measurements only ever inspect each qubit individually, so that a product state $\otimes_i \rho_i$ is indistinguishable---\emph{under $D$}---from a more general mixed state $\rho$ whose reduced density matrix on qubit $i$ is $\rho_i$ for every $i$.\footnote{We stress that in the PAC formalism, the goal is not necessarily to learn the exact state, but simply to find one that behaves similarly under the specified input distribution of measurements. Thus for measurements of the kind drawn from $D$, learning the product of reduced density states is sufficient.} Yet since this distribution on measurements is fundamentally not very interesting for anything other than product states, we state the results in this section only for product states.

The following technical lemma will be the backbone of our results.

\begin{lemma} \label{lemma:haar_random_single}
	For any single qubit pure state $\ket{\psi}\bra{\psi} = \frac{I + P}{2}$ and mixed state $\rho$:
	\[
	\ex_{E \sim D'}\left[\sgn\left(\trc\left(E \ketbra{\psi}{\psi}\right) - \frac{1}{2}\right) \left(\trc\left(E \rho\right) - \frac{1}{2}\right)\right] = \frac{1}{4}\trc(P \rho).
	\]

\end{lemma}
\begin{proof}
	We will decompose $\rho = \lambda \ketbra{\phi}{\phi} + (1-\lambda)\ketbra{\phi^\perp}{\phi^\perp}$ such that $\ket{\phi} = \cos \theta' \ket{\psi} + \sin \theta' \ket{\psi^\perp}$ and $\ket{\phi^\perp} = e^{i \phi'}(\sin \theta' \ket{\psi} - \cos \theta' \ket{\psi^\perp})$. The following identity will be useful at the end:
	
	\begin{align*}
		\trc(P\rho) &= 2\trc(\ketbra{\psi}{\psi}\rho) - 1\\
		&= 2\big[\lambda \trc(\ketbra{\psi}{\psi}\ketbra{\phi}{\phi}) + (1-\lambda)\trc(\ketbra{\psi}{\psi}\ketbra{\phi^\perp}{\phi^\perp})\big] - 1\\
		&= 2\cos^2\theta' \lambda + 2(1-\lambda)\sin^2\theta' - (\sin^2 \theta' + \cos^2 \theta')\\
		&= (2\lambda - 1) \cos^2 \theta' - (2\lambda - 1)\sin^2 \theta'\\
		&= (2\lambda - 1)\cos 2\theta'
	\end{align*}

	Let $U$ be the unitary such that $U\ket{0} = \ket{\psi}$ and $U\ket{1} = \ket{\psi^\perp}$. Due to symmetry, we can parameterize a Haar-random single qubit state using spherical coordinates as $E = \frac{1}{2}U(I + \cos\phi \sin \theta X + \sin \phi \sin \theta Y + \cos \theta Z)U^\dagger$ for the Pauli matrices $X$, $Y$, and $Z$.
	
	\begin{align*}
		\trc\left(E \ketbra{\psi}{\psi}\right)
		&= \trc\left(\frac{1}{2}U(I + \cos\phi \sin \theta X + \sin \phi \sin \theta Y + \cos \theta Z)U^\dagger \ketbra{\psi}{\psi}\right)\\
		&= \trc\left(\frac{1}{2}(I + \cos\phi \sin \theta X + \sin \phi \sin \theta Y + \cos \theta Z) \ketbra{0}{0}\right)\\
		&= \frac{1 + \cos \theta}{2}
	\end{align*}
	
	We can also do the same thing for $\rho$:
	\begin{align*}
		&\trc\left(E \rho\right)\\
		&= \lambda \trc(E \ketbra{\phi}{\phi}) + (1-\lambda)\trc(E \ketbra{\phi^\perp}{\phi^\perp})\\
		&= \lambda\frac{1 + \cos \theta \cos 2\theta' + \cos \phi \sin \theta \sin 2\theta'}{2} + (1-\lambda)\frac{1 - \cos \theta \cos 2\theta' - \cos \phi \sin \theta \sin 2\theta'}{2}\\
		&= \frac{1 + (2\lambda - 1) (\cos \theta \cos 2\theta' + \cos \phi \sin \theta \sin 2\theta')}{2}
	\end{align*}
	
	This allows us to perform a spherical integral over $\theta$ and $\phi$ to get the expectation:
	
	\begin{align*}
		&\ex_{E \sim D'}\left[\sgn\left(\trc\left(E \ketbra{\psi}{\psi}\right) - \frac{1}{2}\right) \left(\trc\left(E \ketbra{\phi}{\phi}\right) - \frac{1}{2}\right)\right]\\
		&= \frac{1}{4\pi}\int_0^{2\pi} d\phi \int_0^\pi d\theta \sin \theta \sgn\left(\cos \theta\right) \left((2\lambda - 1)\frac{\cos \theta \cos 2\theta' + \cos \phi \sin \theta \sin 2\theta'}{2}\right)\\
		&= \frac{2\lambda - 1}{8\pi}\int_0^{2\pi} d\phi \left[\int_0^{\pi/2} -  \int_{\pi/2}^{\pi}\right] d\theta \sin \theta  \left(\cos \theta \cos 2\theta' + \cos \phi \sin \theta \sin 2\theta'\right)\\
		&= (2\lambda - 1)\frac{\pi \cos 2 \theta' + \pi \cos 2 \theta'}{8 \pi}\\
		&= (2\lambda - 1)\frac{\cos 2 \theta'}{4}\\
		&= \frac{1}{4}\trc(P \rho)
	\end{align*}
\end{proof}

Our algorithm for learning product states will be work by learning each qubit in the Pauli basis. This results in a $3n$-query algorithm, corresponding to the $3n$ parameters that it takes to define a product state. We first recall the definition of trace distance, which is the quantum generalization of total variational distance.

\begin{definition}[Trace distance]
	Given quantum states $\rho$ and $\sigma$,
	\[
	\lVert \rho - \sigma \rVert_{\trc} = \frac{1}{2}\sum_i \lvert \lambda_i \rvert
	\]
	where $\{\lambda_i\}$ is the set of eigenvalues of the matrix $\rho - \sigma$.
\end{definition}

\begin{proposition}[folklore]\label{prop:trace}
	Given two single qubit states $\rho$ and $\sigma$, the $\lVert \rho - \sigma \rVert_{\trc}$ is half the Euclidean distance between their points on the Bloch sphere.
\end{proposition}

The following lemma will then be necessary to relate trace distance of the states to the squared loss in learning under this distribution.
\begin{lemma}\label{lemma:error_product_state}
	For $n$-qubit product states $\rho = \bigotimes_i \rho_i$ and $\sigma = \bigotimes_i \sigma_i$, let $f_\rho(E) = 2\trc(E\rho) - 1$ and $f_\sigma(E) = 2\trc(E \sigma) - 1$. Let $D$ be the distribution over measurements defined earlier. Then
	\[
	\ex_{E \sim D}[(f_\rho(E) - f_\sigma(E))^2] = \frac{4}{3n}\sum_{i=1}^n \lVert \rho_i - \sigma_i\rVert_{\trc}^2
	\]
\end{lemma}
\begin{proof}
	Let $\xi = \rho - \sigma$. Then by linearity
	\[
	f_\rho(E) - f_\sigma(E) = 2(\trc(E \rho) - \trc(E \sigma)) = 2\trc(E\xi).
	\]
	
	We will define $\xi_i = \trc_i(\xi) = \rho_i - \sigma_i$ to be the reduced density matrix on the $i\th$ qubit of $\xi$. Noting that each $\xi_i$ is traceless, then by diagonalizing we can write $\xi_i = \lambda_i \ketbra{\lambda_i}{\lambda_i} - \lambda_i \ketbra{\lambda_i^\perp}{\lambda_i^\perp}$ for $\lambda_i \in [0, 1]$ such that $\lambda_i = \lVert \rho_i - \sigma_i \rVert_{\trc}$ is the trace distance of the reduced density matrix.
	
	Like in \cref{lemma:haar_random_single}, we can parameterize a single-qubit Haar-random projection as $E = \frac{1}{2}U(I + \cos\phi \sin \theta X + \sin \phi \sin \theta Y + \cos \theta Z)U^\dagger$, where $U\ket{0} = U\ket{\lambda_i}$ and $U\ket{1} = U\ket{\lambda_i^\perp}$. This implies that $U\xi_iU^\dagger = \lambda_i Z$.
	
	\begin{align*}
	\trc(E \xi_i)
	&=\trc\bigg(\frac{1}{2}U\big(I + \cos\phi \sin \theta X + \sin \phi \sin \theta Y + \cos \theta Z\big)U^\dagger \rho\bigg)\\
	&=\trc\bigg(\frac{1}{2}\big(I + \cos\phi \sin \theta X + \sin \phi \sin \theta Y + \cos \theta Z\big)\cdot \lambda_i Z\bigg)\\
	&= \lambda_i \cos \theta
	\end{align*}
	
	Using this, we now compute the squared-loss as follow.
	
	\begin{align*}
		\ex_{E \sim D}[(f_\rho(E) - f_\sigma(E))^2]
		&= \frac{1}{n}\sum_{i=1}^n \ex_{E' \sim D'}[(f_{\rho_i}(E') - f_{\sigma_i}(E')^2]\\
		&= \frac{4}{n}\sum_{i=1}^n \ex_{E' \sim D'}[\trc^2(E'\xi_i)]\\
		&= \frac{4}{n}\sum_{i=1}^n \frac{1}{2}\int_{0}^\pi \sin \theta \cdot \lambda_i^2 \cos^2 \theta\\
		&= \frac{4}{3n}\sum_{i=1}^n \lambda_i^2\\
		&= \frac{4}{3n}\sum_{i=1}^n \lVert \rho_i - \sigma_i\rVert_{\trc}^2
	\end{align*}
\end{proof}

We now show how to use \cref{lemma:haar_random_single} to learn each qubit of the product state, allowing us to then apply \cref{lemma:error_product_state} to get our learning result.

\begin{theorem}
	Let $D$ be the distribution on measurements and let $\cC$ be the concept class of product states defined earlier. There is an SQ learner that is able to learn $\cC$ under $D$ up to squared loss $\epsilon$ using $3n$ queries of tolerance $\sqrt{\epsilon}/n$.
\end{theorem}
\begin{proof}
	Let the unknown $\rho \in \cC$ be given by $\rho = \bigotimes_i \rho_i$. If we define $P_1 = X$, $P_2 = Y$, and $P_3 = Z$, then our queries will be
	\[
	\phi_{i, j}(E, Y) = \sgn\left(\frac{1}{2^{n-1}}\trc\left(E\cdot \left(I^{\otimes i-1}\otimes \frac{I + P_j}{2} \otimes I^{\otimes n-i}\right)\right) - \frac{1}{2}\right)\cdot Y
	\] The query $\phi_{i, j}$ will correspond to taking the projection of the $i\th$ qubit along the Pauli $P_j$, as we now show:
	\begin{align*}
	\ex_{E \sim D, Y \sim f_\rho(E)} \left[\phi_{i, j}\right(E, Y\left)\right]
	&= \ex_{E \sim D} \left[\phi_{i,j}(E, 1) \trc\left(E \rho\right) + \phi_{i, j}(E, -1)\bigg(1-\trc\left(E \rho\right)\bigg)\right]\\
	&= \ex_{E \sim D} \left[\phi_{i,j}(E, 1) \bigg(2\trc\left(E \rho\right) - 1\bigg)\right]\\
	&= \frac{1}{n} \ex_{E' \sim D'}\left[\sgn\left(\trc\left(E'\ketbra{0}{0}\right) - \frac{1}{2}\right) \bigg(2\trc\left(E' \rho_i \right) - 1\bigg)\right]\\
	&= \frac{1}{2n}\trc(P_j \rho_i).
	\end{align*} Here the third equality exploits the definition of $D$ as $\frac{1}{n} \sum_{i=1}^n I^{\otimes i - 1} \otimes D' \otimes I^{\otimes n-i}$ (only the $i\th$ term yields a nonzero expectation), and the fourth equality is \cref{lemma:haar_random_single}.

	Any specific qubit $\rho_i$ can be written in Bloch sphere coordinates as $\frac{I + x_i X + y_i Y + z_i Z}{2}$. We can estimate $x_i = \frac{1}{2}\trc(P_1 \rho_i)$ up to error $\sqrt{\epsilon}$ using a single query of tolerance $\sqrt{\epsilon}/n$. The same holds true for $y_i$ and $z_i$. If we use this to construct our estimate
	\[
	\hat{\rho}_i = \frac{I + \hat{x}_i X + \hat{y}_i Y + \hat{z}_i Z}{2}
	\]
	then by \cref{prop:trace} we get \[\lVert \rho_i - \hat{\rho_i} \rVert_{\trc}^2 = \frac{1}{4}\left[(x_i - \hat{x}_i)^2 + (y_i - \hat{y}_i)^2 + (z_i - \hat{z}_i)^2\right] \leq 3\epsilon/4. \] Finally, using  \cref{lemma:error_product_state}:
	
	\[
		\ex_{E \sim D}[(f_\rho(E) - f_\sigma(E))^2] \leq \frac{4}{3n}\sum_{i=1}^n 3\epsilon/4 = \epsilon.
	\]
	
	We note that if the estimated point is outside of the Bloch sphere, we can simply normalize the point to the surface of the Bloch sphere and this will never increase the error. To quickly sketch the proof of this, take the plane formed by the center of the sphere, the estimated point $\hat{p}$ that is outside of the sphere, and the real point $p$ which is both within the Bloch sphere and within a sphere $\epsilon$ radius located at $\hat{p}$. The normalized point $\hat{p}'$ is always located on the line from the $\hat{p}$ to the origin, and one can make a separating plane that bisects the line segment between $\hat{p}$ and $\hat{p}'$ that denotes whether one is closer to $\hat{p}$ or $\hat{p}'$. Since the Bloch sphere will always be on the side closer to $\hat{p}'$ and the real point $p$ is in the Bloch sphere, $p$ will always be closer to $\hat{p}'$ than $\hat{p}$.
\end{proof}

We can simplify this algorithm if we know in advance that $\rho$ is a computational basis state. In that case, we know that each qubit $\rho_i$ is either $(I + Z)/2$ or $(I - Z)/2$, and so we only need to make $n$ queries $\phi_{i, 3}$, one for each $i$. Moreover, we only need to identify the coordinate $z_i$ to within an accuracy of $1$ in order to distinguish the $z_i = 1$ and $z_i = -1$ cases, so that our tolerance need only scale as $O(1/n)$ in order to learn $\rho$ perfectly (i.e.\ with $\epsilon = 0$).

\section{Connections to differential privacy}\label{sec:diff-priv}
A PAC learning algorithm $L$ can be viewed as a randomized algorithm that takes as input a training dataset (i.e.\ a set of labeled examples $(x, y)$ sampled from a distribution) and outputs a hypothesis that with high probability has low error over the distribution. That is, if $S$ is a training dataset, then $L(S)$ describes a probability distribution over hypotheses (where the randomness arises from the internal randomness of the learner). Intuitively, differential privacy requires $L$ to satisfy a kind of stability: on any two inputs $S$ and $S'$ that are close, the distributions $L(S)$ and $L(S')$ must be close as well.

\begin{definition}[Differential privacy, \cite{dwork2014algorithmic}]
	Call two datasets $S = \{(x_i, y_i)\}_{i=1}^m$ and $S' = \{(x_i', y_i')\}_{i=1}^m$ neighbors if they only differ in one entry. A learner $L$ (understood in the sense just discussed) is said to be $\alpha$-differentially private (or $\alpha$-DP for short) if for any $S$ and $S'$ that are neighbors, the distributions $L(S)$ and $L(S')$ are close in the sense that for any hypothesis $h$, $\pr[L(S) = h] \leq e^{\alpha} \pr[L(S') = h]$.
\end{definition}

A well-known property of SQ algorithms is that they can readily be made differentially private \cite{blum2005practical, dwork2014algorithmic}. Since differential privacy is a notion that is well-defined only in the PAC setting where the input is a set of training examples (as opposed to access to an SQ oracle), such a statement is necessarily of the form ``any SQ learner yields a PAC learner that satisfies differential privacy.''

\begin{theorem}[see e.g.~\cite{balcan2015dp}]\label{thm:sq_dp}
	Let $\cC$ be a concept class learnable up to error $\epsilon$ by an SQ learner $L$ using $q$ queries of tolerance $\tau$. Then it is also learnable up to error $\epsilon$ in the PAC setting by an $\alpha$-DP learner $L'$ with sample complexity $\tilde{O}(\frac{q}{\alpha \tau} + \frac{q}{\tau^2})$ (with constant probability).
\end{theorem}

The proof is standard and proceeds by simulating each of $L$'s queries using empirical estimates over a sample of size roughly $1/\tau^2$ and then using the Laplace mechanism to add some further noise.

One can extend this notion to the quantum setting. One natural and direct way of doing so is simply by replacing the classical dataset of labeled pairs $(x_i, y_i)$ by one of measurement-outcome pairs $(E_i, Y_i)$; the rest remains exactly analogous. \cref{thm:sq_dp} then carries over verbatim to our notion of quantum SQ learnability. This form of quantum differential privacy was recently studied by \citet{arunachalam2021private}, who were able to relate it to online learning, one-way communication complexity, and shadow tomography of quantum states, extending ideas of \citet{bun2020equivalence}. Since our notion of quantum SQ learnability implies quantum DP learnability, it also fits into their framework. In particular, by the chain of implications established in that work, efficient quantum SQ learnability of a class of states implies DP PAC learnability, which implies finite sequential fat-shattering (sfat) dimension, which in turn implies online learnability, gentle shadow tomography, and ``quantum stability.'' In fact, in the classical setting, some of the main examples of realistic DP learners are SQ (even though technically the inclusion is known to be strict) \cite{blum2005practical, kasiviswanathan2011can}, and one might expect the same to hold in the quantum setting as well.

We remark that a somewhat different kind of quantum differential privacy, where privacy is with respect to copies of the unknown state, may also be defined as follows. View a quantum state learner $L$ as an algorithm that takes in multiple copies $\rho^{\otimes m}$ of some unknown state $\rho$, is allowed to sample and perform random measurements from a distribution $D$, and outputs another state $\sigma$ that is required to be close to $\rho$ with respect to $D$ with high probability. If the random measurements are viewed as the internal randomness of the learner, then this is similar to the view we took of a classical learner earlier. We can now define a notion of differential privacy for quantum state learners by requiring that $L(\rho^{\otimes m})$ and $L(\rho^{\otimes m-1} \otimes \rho')$ (where $\rho \neq \rho'$, so that $\rho^{\otimes m}$ and $\rho^{\otimes m-1} \otimes \rho'$ are neighbors) are $\alpha$-close as distributions over states (in the natural way). This can also be seen as a stylized kind of tolerance to noise or corruptions. The following analogue of \cref{thm:sq_dp} can then be proven using almost exactly the same proof; essentially, we are only replacing classical examples with copies of quantum states.

\begin{theorem}
	Let $\cC$ be a class of quantum states learnable up to error $\epsilon$ by an SQ learner $L$ using $q$ queries of tolerance $\tau$. Then it is also learnable up to error $\epsilon$ in the PAC setting by an $\alpha$-DP learner $L'$ (in the specific sense just described) with copy complexity $\tilde{O}(\frac{q}{\alpha \tau} + \frac{q}{\tau^2})$ (with constant probability).
\end{theorem}

Note that these notions are different from those of \cite{aaronson2019gentle}, which defined differential privacy for quantum measurements. Here two $n$-qubit states are considered neighbors if it is possible to reach one from the other by a quantum operation (sometimes called a superoperator) on a single qubit. In particular, two product states $\rho = \otimes_i \rho_i$ and $\sigma = \otimes_i \sigma_i$ are neighbors if $\rho_i = \sigma_i$ for all $i$ but one.

\begin{definition}[Quantum differential privacy for measurements, \cite{aaronson2019gentle}]
	A measurement $M$ is said to be $\alpha$-DP if for any $n$-qubit neighbor states $\rho, \sigma$, and any outcome $y$, $\pr[M(\rho) = y] \leq e^{\alpha} \pr[M(\sigma) = y]$.
\end{definition}

The authors show that this definition can be related to the notion of a ``gentle quantum measurement,'' and this connection can be carefully exploited to perform shadow tomography \cite{aaronson2019shadow}. However, this kind of quantum DP is not applicable in a natural way to a PAC or SQ learner, since such a learner is an algorithm rather than just a single measurement.

\section{Discussion and open problems}

\paragraph{Statistical vs.\ query complexity.} Conceptually, the contrast between our SQ model and the original PAC model of \cite{aaronson2007learnability} is interesting. Apart from the definition of an elegant model, Aaronson's main insight was to characterize learnability in a purely \emph{statistical} sense, showing bounds on sample complexity via an analysis of the so-called fat-shattering dimension of quantum states. In learning theoretic terms, this took advantage of a separation of concerns that the PAC model encourages: (a) empirical performance, i.e.\ a learner achieving low error with respect to the training data, and (b) generalization, i.e.\ this performance actually generalizing to the true distribution. The SQ model, however, does not naturally accommodate such a separation. SQ algorithms are instead primarily characterized by the number of queries required; generalization is ``in-built.'' The closest analogue to a notion of sample complexity is the role played by the tolerance, and the closest thing to studying generalization on its own might have been to show a phase transition in what different regimes of the tolerance are able to accomplish. The formal statements of our SQ lower bounds do have such a flavor: ``\emph{either} use small tolerance \emph{or} many queries.''

\paragraph{Suitable classes and distributions for PAC-learning.} It is notable that the algorithms of \cite{rocchetto2018stabiliser} for learning stabilizer states and \cite{yoganathan2019condition} for low Schmidt rank states are essentially the only known positive results in the framework of \cite{aaronson2007learnability}. Both these algorithms are ``algebraic'' and involve solving a system of polynomial equations, something that SQ cannot handle. A longstanding question in this area is: what other interesting classes can be learned, and under what distributions on measurements? And can they also be learned in the SQ setting?

A major issue in picking suitable distributions on measurements is that under many natural distributions, the maximally mixed state actually performs quite well, so that the problem of learning becomes essentially superfluous. Even in this work, we obtained lower bounds for learning stabilizer states under the uniform distribution on Pauli measurements only for learning up to exponentially small squared loss. This was because the norms of the p-concepts are themselves exponentially small, or in other words the maximally mixed state already achieves exponentially small loss. We were able to get around this and obtain a $\Omega(2^n)$ lower bound via a direct reduction from learning parities (by considering parity measurements). Can we do better than just $2^n$? Is there a $\omega(2^n)$-sized (e.g., $4^n$ or $2^{n^2}$) subset of stabilizer states such that there exists a distribution over Pauli measurements inducing norms that are only polynomially small yet have an exponentially small average correlation? That is, is there a $\omega(2^n)$-sized set of stabilizer states and accompanying distribution over Pauli measurements such that the maximally mixed state does not do well?

\paragraph*{Other forms of noise.} Can we extend the noise tolerance of SQ algorithms to more forms of noise, or improve the parameters of the noise tolerated? One such interesting form of noise would be depolarizing noise that acts on individual qubits (as opposed to acting directly on the whole state).

\paragraph*{Noise-tolerant learning beyond SQ.} The best-known PAC algorithm for learning parities with noise is due to \cite{blum2003noise} and runs in slightly subexponential time. Interestingly, this means it beats the exponential SQ lower bound and is hence essentially the only known example of a noise-tolerant PAC algorithm that is not SQ (although it cannot handle noise arbitrarily close to the information-theoretic limit). Can we similarly hope for a noise-tolerant but non-SQ learner for stabilizer states that runs in subexponential time?

\section*{Acknowledgements}
We thank Scott Aaronson, Srinivasan Arunachalam, and Andrea Rocchetto for many helpful discussions. AG was supported by NSF awards AF-1909204, AF-1717896, and the NSF AI Institute for Foundations of Machine Learning (IFML). DL was supported by the Simons It from Qubit Collaboration and Scott Aaronson's Vannevar Bush Faculty Fellowship from the US Department of Defense.

\bibliographystyle{plainnat}
\bibliography{../refs-zotero}

\end{document}